\documentclass[11pt]{article}

\usepackage[margin=2.5cm]{geometry}
\usepackage{authblk}

\usepackage{amssymb,amsmath,amsthm,amsfonts,thmtools}
\usepackage{color}
\usepackage[export]{adjustbox}
\usepackage{yfonts}
\usepackage{mathrsfs}

\usepackage{graphicx}
\usepackage{xspace}

\usepackage{times}
\usepackage{verbatim}
\usepackage{enumitem}
\usepackage{graphicx,wrapfig,lipsum}
\usepackage[font=small]{caption}
\usepackage{subcaption}
\usepackage[english]{babel}
\usepackage[utf8]{inputenc}

\usepackage[noend]{algpseudocode}
\usepackage{url}
\usepackage[all]{xy}
\usepackage{rotating}
\usepackage{ifpdf}
\usepackage{tcolorbox}
\usepackage{lipsum}

\usepackage{mdframed}             
\usepackage{xifthen} 
\usepackage{mathtools}  

\DeclareMathAlphabet\mathbfcal{OMS}{cmsy}{b}{n}

\usepackage{comment}
\usepackage[T1]{fontenc}

\usepackage[hidelinks]{hyperref} 

\usepackage{array}
\usepackage[title]{appendix}
\usepackage{graphicx}
\usepackage{tikz}
\usetikzlibrary{automata,positioning} \usepackage{dsfont} \usepackage{xspace}
\usepackage[nothing]{algorithm}
\usepackage{algorithmicx}
\usepackage[noend]{algpseudocode}
\algblockdefx[protocol]{Protocol}{EndProtocol}{\textbf{Protocol}\ }{}

\floatname{algorithm}{Pseudocode}

\newcommand{\barz}{{\bar z}}

\newcommand{\calA}{\ensuremath{\mathcal A}\xspace}

\newcommand{\tildet}{{\tilde{t}}} 

\newcommand{\tildeO}{{\tilde{O}}}

\colorlet{darkgreen}{green!45!black}

\newtheorem{theorem}{Theorem}
\newtheorem{lemma}{Lemma}[section]

\newtheorem{claim}[lemma]{Claim}

\newcommand{\ignore}[1]{}

\newcommand{\margincomment}[2]%
{\marginpar{\footnotesize\raggedright {\color{red}#1}: #2}}

\newcommand{\etal}{et~al.}

\newcommand{\myparagraph}[1]{{\medskip\noindent\textbf{#1}}}
\newcommand{\emparagraph}[1]{{\medskip\noindent\textit{#1}}}

\newcommand{\mycase}[1]{{\underline{Case~#1}:}}

\newcommand{\braced}[1]{{ \left\{ {#1} \right\} }}

\newcommand{\barred}[1]{{ \left| {#1} \right| }}

\newcommand{\assign}{\,{\leftarrow}\,} 

\newcommand{\integers}{{\mathbb Z}}

\newcommand{\Exp}{{\mbox{\rm Exp}}}

\newcommand{\weight}{\ell}
\newcommand{\maxedgeweight}{{\weight}_{\textrm{max}}}

\newcommand{\algA}{\mathcal{A}}
\newcommand{\algB}{\mathcal{B}}
\newcommand{\algR}{\mathcal{R}}
\newcommand{\algorithms}{{\mathbb{A}}}
\newcommand{\randalgR}{\mathcal{R}}
\newcommand{\weightset}{{\mathbb{L}}}
\newcommand{\runtime}{{T}}

\newcommand{\phaselength}{{\partial{t}}}

\newcommand{\modelQ}{{\mathbb{Q}}}

\newcommand{\half}{\textstyle{\frac{1}{2}}}
\newcommand{\onethird}{\textstyle{\frac{1}{3}}}
\newcommand{\onesixth}{\textstyle{\frac{1}{6}}}
\newcommand{\onetwelth}{\textstyle{\frac{1}{12}}}

\newcommand{\ListLengths}{\setlength{\itemsep}{0ex}\setlength{\topsep}{1ex}\setlength{\partopsep}{0ex}}

\graphicspath{{./}}

\title{Lower Bounds for Adaptive Relaxation-Based Algorithms for Single-Source Shortest Paths\thanks{Research supported by NSF grant CCF-2153723.}}

\author[1]{Sunny Atalig}
\author[1]{Alexander Hickerson}
\author[1]{Arrdya Srivastav}
\author[2]{Tingting Zheng}
\author[1]{Marek Chrobak}

\affil[1]{University of California at Riverside, USA}
\affil[2]{Guangdong University of Technology, China}

\begin{document}

\maketitle

\begin{abstract}
We consider the classical single-source shortest path problem in directed weighted graphs.
D.~Eppstein proved recently an $\Omega(n^3)$ lower bound for
oblivious algorithms that use relaxation operations to update the
tentative distances from the source vertex.  We generalize this result
by extending this $\Omega(n^3)$ lower bound to \emph{adaptive} algorithms that, in addition to relaxations, can perform 
queries involving some simple types of linear inequalities between edge
weights and tentative distances. Our model captures as a special case the 
operations on tentative distances used by Dijkstra's algorithm.
\end{abstract}


\section{Introduction}
\label{sec: introduction}


We consider the classical single-source shortest path problem in directed weighted graphs.
In the case when all edge weights are non-negative, Dijkstra's algorithm~\cite{1959_dijkstra_algorithm}, if implemented
using Fibonacci heaps, computes the shortest paths in time $O(m + n\log n)$, where $n$ is the
number of vertices and $m$ is the number of edges. In the general case,
when negative weights are allowed (but not negative cycles),
the Bellman-Ford algorithm~\cite{1955_shimbel_structure_in_communication_nets,1958_bellman_routing_problem,1959_moore_shortest_path,1956_ford_network_flow_theory}
solves this problem in time $O(nm)$. 

Both algorithms work by repeatedly executing operations of \emph{relaxations}. 
(This type of algorithms are also sometimes called label-setting algorithms~\cite{1984_deo_pang_shortest_path}.)
Let $\weight_{uv}$ denote the weight of an edge $(u,v)$.
For each vertex $v$, these algorithms maintain a value $D[v]$
(that we will refer to as the \emph{D-value} at $v$)
that represents the current upper bound on the distance from the source vertex $s$ to $v$. 
A relaxation operation for an edge $(u,v)$
replaces $D[v]$ by $\min \braced{D[v] , D[u]+\ell_{uv}}$. That is, $D[v]$ is replaced
by $D[u] + \weight_{uv}$ if visiting $v$ via $u$ turns out to give a shorter distance
to $v$, based on the current distance estimates.
When the algorithm completes, each value $D[v]$ is equal to the correct distance from $s$ to $v$.
Dijkstra's algorithm executes only one relaxation for each edge, while
in the Bellman-Ford algorithm each edge can be relaxed $\Theta(n)$ times.

We focus on the case of complete directed graphs, in which case $m = n(n-1)$.
For complete graphs, the number of relaxations in Dijkstra's algorithm is $\Theta(n^2)$.
In contrast, the Bellman-Ford algorithm executes $\Theta(n^3)$ relaxations.
This raises the following natural question: \emph{is it possible to
solve the shortest-path problem by using asymptotically fewer than $O(n^3)$ relaxations,
even if negative weights are allowed?}

To make this question meaningful, some restrictions need to be imposed on allowed algorithms.
Otherwise, an algorithm can ``cheat'': it can compute the shortest paths without any explicit use of
relaxations, and then execute $n-1$ relaxations on the edges in  the shortest-path tree,
in order of their hop-distance from $s$, thus making only $n-1$ relaxations.

Eppstein~\cite{2023_Eppstein_non-adaptive_shortest_path} circumvented this issue
by assuming a model where the sequence of relaxations is independent of the weight assignment. 
Then the question is whether there is a short ``universal'' sequence of relaxations,
namely one that works for an arbitrary weight assignment. 
The Bellman-Ford algorithm is essentially such a universal sequence of length $O(n^3)$.
Eppstein~\cite{2023_Eppstein_non-adaptive_shortest_path} proved that this is asymptotically
best possible; that is, each universal relaxation sequence must have $\Omega(n^3)$ relaxations.
This lower bound applies even in the randomized case, when the
relaxation sequence is generated randomly and the objective is to minimize the
expected number of relaxations.

The question left open in~\cite{2023_Eppstein_non-adaptive_shortest_path} is  
whether the $\Omega(n^3)$ lower bound applies to relaxation-based \emph{adaptive} algorithms,
that generate relaxations based on information collected during the computation.
(This problem is also mentioned by Hu and Kozma~\cite{2024_Hu_Kozma_Yens_improvement_is_optimal}, who
remark that lower bounds for adaptive algorithms have been ``elusive''.)
We answer this question in the affirmative for some natural types of adaptive algorithms.

In our computation model, an algorithm is allowed to perform two types of operations:
(i) \emph{queries}, which are simple linear inequalities involving edge weights and D-values,
and (ii) \emph{relaxation updates}, that modify D-values.
The action at each step depends on the outcomes of the earlier executed queries.
Such algorithms can be represented as decision trees, with
queries and updates in their nodes, and with each query node having two
children, one corresponding to the ``yes'' outcome and the other to the ``no'' outcome.

Specifically, we study \emph{query/relaxation-based algorithms} that can make queries of three types:
\begin{description}\setlength{\itemsep}{-0.02in}
\item{\emph{D-comparison query:}} ``$D[u] < D[v]$?'', for two vertices $u$, $v$,
\item{\emph{Weight-comparison query:}} ``$\weight_{uv} < \weight_{xy}$?'', for two edges $(u,v)$, $(x,y)$,
\item{\emph{Edge query:}} ``$D[u] + \weight_{uv} < D[v]$?'', for an edge $(u,v)$,
\end{description}
and can update D-values as follows:
\begin{description}\setlength{\itemsep}{-0.02in}
\item{\emph{Relaxation update:}} ``$D[v] \assign \min\braced{D[v], D[u] + \weight_{uv}}$'', for an edge $(u,v)$.
\end{description}

Throughout the paper, for brevity, we will write ``D-query'' instead of ``D-comparison query''
and ``weight query'' instead of ``weight-comparison query''.

We assume that initially $D[s] = 0$ and $D[v] = \ell_{sv}$ for all vertices $v \ne s$. 
This initialization and the form of relaxation updates ensure that at all times each value $D[v]$
represents the length of some simple path from $s$ to $v$.
Thus D-queries and edge queries amount to comparing the lengths of two paths from $s$.
Further, the D-values induce a tentative approximation of the shortest-path tree,
where a node $u$ is the parent of a node $v$ if the last decrease of $D[v]$ 
resulted from a relaxation of edge
$(u,v)$. So the algorithm's decision at each step depends on this tentative shortest-path tree.


\myparagraph{Our contributions.}
We start by considering algorithms that use only edge queries. For such algorithms
we prove the following $\Omega(n^3)$ lower bound:


\begin{theorem}\label{thm: relaxation-outcome lower bound} 
(a) Let $\algA$ be a deterministic query/relaxation-based algorithm for the single-source shortest path problem that uses only
edge queries. Then the running time of $\algA$ is $\Omega(n^3)$, even if
the weights are non-negative and symmetric (that is, the graph is undirected).
(b) If $\algA$ is a randomized algorithm then the same $\Omega(n^3)$ lower bound holds for $\algA$'s expected running time.
\end{theorem}


We first give the proof of Theorem~\ref{thm: relaxation-outcome lower bound}(a), the lower
bound for deterministic algorithms. In this proof, (in Section~\ref{sec: deterministic algorithms relaxation queries}), 
we view the computation of $\algA$ as a game against an adversary who gradually constructs a weight assignment,
consistent with the queries,
on which most of the edge queries performed by $\algA$ will have negative outcomes, thus revealing
little information to $\algA$ about the structure of the shortest-path tree.

Then, in Section~\ref{sec: deterministic algorithms three types of queries}, we
show how to extend this lower bound to all three types of queries if negative weights are allowed, 
proving Theorem~\ref{thm: all queries lower bound}(a) below.


\begin{theorem}\label{thm: all queries lower bound} 
(a) Let $\algA$ be a deterministic query/relaxation algorithm for the single-source shortest path problem that uses the
three types of queries: $D$-queries, weight-queries, and edge queries, as well as relaxation updates.
Then the running time of $\algA$ is $\Omega(n^3)$.
(b) If $\algA$ is a randomized algorithm then the same $\Omega(n^3)$ lower bound holds for $\algA$'s expected running time.
\end{theorem}


Our query/relaxation model captures as a special case the operations on
tentative distances used by Dijkstra's algorithm, because D-queries are sufficient to maintain 
the ordering of vertices according to their D-values. More broadly,
Theorem~\ref{thm: all queries lower bound} may be helpful in guiding future research on 
speeding up shortest-path algorithms for the general case, when negative weights are allowed,
by showing limitations of na\"{i}ve approaches based on extending Dijkstra's algorithm.

\smallskip

The proof of Theorem~\ref{thm: all queries lower bound}(a) is essentially via a reduction, 
showing that the model with all three types of queries can
be reduced to the one with only edge queries, and then applying the lower bound from Theorem~\ref{thm: relaxation-outcome lower bound}(a).
This reduction modifies the weight assignment, making it asymmetric and introducing negative weights.
As a side result, we also observe in Theorem~\ref{thm: arbitrary graph} that
this reduction works even for arbitrary (not necessarily complete) graphs,
giving a lower bound that generalizes the one in~\cite{2023_Eppstein_non-adaptive_shortest_path},
as it applies to adaptive algorithms in our query/relaxation model.

Finally, in Section~\ref{sec: lower bounds for randomized algorithms}
we extend both lower bounds to randomized algorithms.  The proofs are based on Yao's
principle~\cite{1977_yao_probabilistic_computations}; that is, we give a probability
distribution on weight assignments on which any deterministic algorithm performs poorly.

Our lower bound results are valid even if all weights are integers of polynomial size.
In the proof of Theorem~\ref{thm: relaxation-outcome lower bound} all weights
are non-negative integers with maximum value $\maxedgeweight = O(n)$.
The proof of Theorem~\ref{thm: all queries lower bound} uses
Golomb rulers~\cite{1938_singer_projective_geometry,1941_erdos_sidon,2005_bollobas_differences}
(also known as Sidon sets) to construct weight assignments with maximum value $\maxedgeweight= O(n^4)$.
In the randomized case, these bounds increase by a factor of $O(n)$.

As explained near the end of Section~\ref{sec: lower bounds for randomized algorithms},
the lower bounds for expectation in Theorems~\ref{thm: relaxation-outcome lower bound}(b) and~\ref{thm: all queries lower bound}(b)
can be quite easily extended to high-probability bounds.


\myparagraph{Related work.}
As earlier mentioned, the Bellman-Ford algorithm can be thought of as a universal
relaxation sequence. It consists of $n-1$ iterations with each iteration relaxing all edges
in some pre-determined order, so the length of this sequence is $(1+o(1))n^3$.
The leading constant $1$ in this bound was reduced to $\half$ by Yen~\cite{1975_yen_shortest_path_network_problems}, 
who designed a universal sequence with $(\half + o(1))n^3$ relaxations.
Eppstein's lower bound in~\cite{2023_Eppstein_non-adaptive_shortest_path} shows in fact
a lower bound of $\onesixth$ on the leading constant,
and just recently Hu and Kozma~\cite{2024_Hu_Kozma_Yens_improvement_is_optimal}
proved that constant $\half$ is in fact optimal.

Bannister and Eppstein~\cite{2012_bannister_eppstein_randomized_speedup} showed that the
leading constant can be reduced to $\onethird$ with randomization, namely that there is
a probability distribution on relaxation sequences for which a sequence, drawn from
this distribution, will compute correct distances in expected time $(\onethird + o(1))n^3$
(or even with high probability). 
Eppstein's lower bound proof~\cite{2023_Eppstein_non-adaptive_shortest_path} for randomized
sequences shows that this constant is at least $\onetwelth$. 

Some of the above-mentioned papers extend the results to graphs that are not necessarily complete.
In particular, Eppstein~\cite{2023_Eppstein_non-adaptive_shortest_path} proved that for
$n$-vertex graphs with $m$ edges, $\Omega(mn/ \log n)$ relaxations are necessary.

The average-case complexity of the Bellman-Ford and Dijkstra's algorithms
has also been studied. For example, 
Meyer~\etal~\cite{2011_meyer_etal_new_bounds_for_old_algorithms}
show that the Bellman-Ford algorithm requires $\Omega(n^2)$ steps on average,
if the weights are uniformly distributed random numbers from interval $[0,1]$.

Some work has been done on improving lower and upper bounds in models beyond our query/relaxation setting. 
Of those, the recent breakthrough paper by Fineman~\cite{2024_Fineman_single-source_shortest_paths}
is particularly relevant. It
gives a randomized $\tildeO(mn^{8/9})$-expected-time algorithm for computing single-source
shortest paths with arbitrary weights. Fineman's computation model is not far from ours in the
sense that the weights are arbitrary real numbers and
 the only arithmetic operations on weights are additions and subtractions, but it
also needs branch instructions that cannot be expressed using our queries. 

The special case when weights are integers is natural and has been extensively 
investigated (see~\cite{2017_cohen_etal_negative-weight_shortest_paths,1995_goldberg_scaling_algorithms,2022_bernstein_etal_negative_weights}, for example).
In the integer domain one can extract information about the weight distribution, and thus
about the structure of the shortest-path tree, using
operations other than linear inequalities involving weights.  The state-of-the-art in this model is the 
(randomized) algorithm by Bernstein~\etal~\cite{2022_bernstein_etal_negative_weights}
that achieves running time $O(m\log^8 n \log W)$ with high probability for weight assignments
where the smallest weight is at least $-W$ (and $W\ge 2$).

Some lower bounds have also been reported for related problems, for
example for shortest paths with restrictions on the number of hops~\cite{2004_cheng_ansari_all_hops,2002_guerin_orda_computing_shortest,2023_kociumaka_bellman-ford_optimal}
or $k$-walks~\cite{2016_jukna_schnitger_on_the_optimality}.


\section{Preliminaries}
\label{sec: preliminaries}

The input is a weighted complete directed graph $G$.
The set of all vertices of $G$ is denoted by $V$, and $s\in V$ is designated as the \emph{start vertex}.
The set of all edges of $G$ is denoted by $E$ and a \emph{weight assignment} is a function $\weight \colon E \to \integers$.
(While real-valued weights are common in the literature, in our constructions we only need integers.)
We will use notations $\weight(u,v)$ and $\weight_{uv}$ for the weight of an edge $(u,v)$. 
By $\maxedgeweight$ we denote the maximum absolute value of an edge weight, that is
$\maxedgeweight = \max_{(u,v)\in E}\barred{\ell_{uv}}$.

Whenever we write ``path'' we mean a ''simple path'', that is a path where each vertex is visited at most once.
The \emph{distance from $x$ to $y$} is defined as the length of the shortest path from $x$ to $y$.
We will assume that the input graph does not have negative cycles. Note that this assumption
gives an algorithm additional information that can potentially be used to reduce the running time. 

The edges in the shortest paths from $s$
to all other vertices form a tree that is called the shortest-path tree. The root of this tree is $s$.
(There is a minor subtlety here related to ties. A more precise statement is that
\emph{there is a way to break ties}, so that the shortest paths form a tree.)


\myparagraph{Formalizing query/relaxation models.} 
We now formally define our computation model. We assume that each vertex $v$ has an associated value $D[v]$, called the D-value at $v$.
Initially $D[s] = 0$ and $D[v] = \ell_{sv}$ for $v\neq s$.
A \emph{query} is a boolean function whose arguments are edge weights and D-values. 
A \emph{query model} $\modelQ$ is simply a set of allowed queries.
For example, the model that has only edge queries is $\modelQ= \braced{1_{D[v] < D[u] + \ell_{uv}} \mid (u,v) \in E}$,
where $1_\xi$ is the indicator function for a predicate $\xi$.
The query/relaxation model in Theorem~\ref{thm: relaxation-outcome lower bound} has query model
$\modelQ$ consisting of all D-queries, weight-queries, and edge queries.
(The reduction in Section~\ref{sec: deterministic algorithms three types of queries} is actually for
an even more general query model.)

An algorithm $\calA$ using a query/relaxation model $\modelQ$ is then a decision tree,
where each internal node corresponds to either a query from $\modelQ$ (with one ``yes'' and one ``no'' branch)
or a relaxation operation (which has one branch), and each leaf is a relaxation operation.
(These leaves have no special meaning.) 
With this definition, at any step of the computation, $D[v]$ represents
the length of a path from $s$ to $v$. 
This decision tree must correctly compute all distances from $s$;
that is, for each weight assignment $\weight$, when the computation of $\calA$ reaches a leaf then
for each vertex $v$ the value of $D[v]$ must be equal to the distance from $s$ to $v$.

The running time of $\algA$ for a weight assignment $\weight$ is defined as the number of steps performed
by $\algA$  until each value $D[v]$ is equal to the correct distance from $s$ to $v$. 
Notice that this is not the same as the depth of the decision tree, 
which could be greater. 
(This definition matches the concept of ``reduced cost'' used in~\cite{2023_Eppstein_non-adaptive_shortest_path} 
for non-adaptive algorithms. For deterministic algorithms we could as well define the running time
as the maximum tree depth, but this definition wouldn't work in the randomized case.)


\myparagraph{Edge weights using potential functions.} 
We define a \emph{potential function} as a function $\phi : V\to \integers$ with $\phi(s) = 0$.
(These functions are also sometimes called \emph{price functions} in the literature.)
To reduce clutter, we will sometimes write the potential value on $v$ as $\phi_v$ instead of $\phi(v).$

A potential function induces a weight assignment $\Delta\phi$ defined by $\Delta\phi(u, v) = \phi_v - \phi_u$, for each $(u,v)\in E$.
Such potential-induced weights satisfy the following
\emph{path independence property:}	For any two vertices $u$, $v$,
all paths from $u$ to $v$ have the same length, namely $\Delta\phi(u, v)$.
Note that $\Delta\phi$ will have some negative weights, unless $\phi$ is identically $0$,
but it does not form negative cycles. Also, every spanning tree rooted at $s$ is a shortest-path tree for $\Delta\phi$.

Any weight assignment $\ell$ can be combined with a potential function $\phi$ to
obtain a new weight assignment $\ell' = \ell + \Delta\phi$.
Such $\ell'$ satisfies the following \emph{distance preservation property:} 
For any two vertices $u$, $v$ and any path $P$ from $u$ to $v$, we have ${\ell'}(P) = \ell(P) + \phi_v - \phi_u$.

Due to the above properties, potential functions have played a key role in the most recent single-source shortest path algorithms  \cite{2022_bernstein_etal_negative_weights,2024_Fineman_single-source_shortest_paths},
in particular being used to transform a negative weight assignment into a non-negative one so that Dijkstra's algorithm can be applied.
However, in this paper we will use them for an entirely different purpose, which is to
construct difficult weight assignments in Section~\ref{sec: deterministic algorithms three types of queries}.
Roughly, a potential-induced assignment 
$\Delta\phi$ can act as a ``mask'' on top of existing weights that renders D-queries and weight queries useless.


\section{Lower Bound for Deterministic Algorithms with Edge Queries}
\label{sec: deterministic algorithms relaxation queries}


This section gives the proof of Theorem~\ref{thm: relaxation-outcome lower bound}(a).
That is, we prove that every deterministic
algorithm that uses relaxations and edge queries needs to make  $\Omega(n^3)$ operations to compute
correct distances. This
lower bound applies even if all weights are non-negative and the weight assignment is symmetric.
(One can think of it as an undirected graph, although we emphasize that in the proof below
we use directed edges.)

For the proof, fix an algorithm $\algA$. We will show how to construct a weight assignment such that only
after $\Omega(n^3)$ operations the D-values computed by $\algA$ represent the correct distances from the source vertex.

Each weight assignment considered in our construction is symmetric and is uniquely specified by a permutation of the vertices.
The weight assignment corresponding to a permutation $\pi = x_0,x_1,...,x_{n-1}$, where $x_0 = s$,
is defined as follows:  for any $0 \le i < j < n$,
\begin{equation*}
\weight_\pi(x_i,x_j) \;=\; 
			\begin{cases}
						2	 	& \;\textrm{if}\; j = i+1
						\\
						L-5i/2 &\;\textrm{if}\; j \ge i+2 \;\textrm{and}\; i \;\textrm{is even}
						\\
						L	& \;\textrm{if}\; j \ge i+2 \;\textrm{and}\; i \;\textrm{is odd}
			\end{cases}
\end{equation*}
where $L$ is some sufficiently large integer, say $L = 5n$.
Then the shortest path tree is just a Hamiltonian path $x_0,x_1,...,x_{n-1}$.
Note that the distance between any two vertices is less than $2n$, while each
edge not on this path has length larger than $2n$.

The proof is by showing an adversary strategy that gradually constructs a permutation of the vertices
in response to $\algA$'s operations. 
The strategy consists of $(n-1)/2$ phases. (For simplicity, assume that $n$ is odd.)
When a phase $k$ starts, for $k = 1,...,(n-1)/2$, the adversary will have already revealed a prefix $X_{k-1} = x_0,x_1,...,x_{2k-2}$
of the final permutation. The goal of this phase is to extend $X_{k-1}$ by two more vertices,
responding to $\algA$'s queries and updates so as to force $\algA$ to make as many operations as possible within the phase,
without revealing anything about the rest of the permutation.

To streamline the proof,  we think about the initial state as following the non-existent $0'$th phase,
and we assume that the D-values for all vertices other than $s$ are initialized to $L+1$, instead of $L$.

We now describe the adversary strategy in phase $k$, by specifying how the adversary responds to each operation of $\algA$
executed in this phase.  Let $Y_{k-1} = V\setminus X_{k-1}$,
let $A$ be set of the edges from $x_{2k-2}$ to $Y_{k-1}$ and $B$ be the set of edges inside $Y_{k-1}$.
The adversary will maintain marks on the edges in $A\cup B$, starting with all edges unmarked. 
We will say that $\algA$ \emph{accesses} an edge $(u,v)$ if it executes either an
edge query or a relaxation for $(u,v)$. 

The idea is this: because of the choice of edge weights and the invariants on the D-values (to be presented soon),
each edge query for an edge $(x_{2k-2},y)\in A$ not yet relaxed in this phase will have a positive outcome.
This way, these responses will not reveal what the next vertex $x_{2k-1}$ on the path is. 
The adversary waits until
$\algA$ relaxes all these edges, and keeps track of these relaxations by marking the relaxed edges. 
At the same time, $\algA$ may be accessing edges in $B$.
The adversary waits until the last access of $\algA$ to an edge $(u,v)\in B$ for which edge $(x_{2k-2},u)$
is already marked. Until this point, all queries to edges in $B$ have negative outcomes.
Only this last edge will have a positive outcome to an edge query, if it's made by $\algA$, 
and the adversary will further make sure that this edge gets relaxed, before ending the phase.

To formalize this, let $(u,v)$ be the edge accessed by $\algA$ in the current operation. We
describe the adversary's response by distinguishing several cases:
\begin{description}\setlength{\itemsep}{-0.03in}
		\item{(s1)} $(u,v) = (x_{2k-2},v)\in A$.
			If this is a relaxation, mark $(u,v)$.
			If this is an edge query do this: if $(u,v)$ is unmarked, respond ``yes'', else respond ``no''.
		\item{(s2)} $(u,v) \in B$. We have two sub-cases depending on the type of access.
		\begin{description}
		\item{Relaxation:} If $(x_{2k-2},u)$ is marked, mark $(u,v)$. 
			If all edges in $A\cup B$ are marked, end phase $k$.
		\item{Edge query:}
			 If $(x_{2k-2},u)$ is not marked, respond ``no''.
			 So suppose that $(x_{2k-2},u)$ is marked.
			 In that case, if $(u,v)$ is not the last unmarked edge in $A\cup B$, mark it and respond ``no''.
			 If $(u,v)$ is the last unmarked edge, respond ``yes'' (without marking).
		\end{description}
		\item{(s3)} $(u,v)\notin A\cup B$. If this is an edge query, respond ``no''.
			If this is a relaxation for $(u,v)$, do nothing.
\end{description}
Let $(u^\ast,v^\ast)$ be the edge marked last in this phase. This is the edge $(u,v) \in B$ from
rule~(s2) that becomes marked when it gets relaxed with all other edges in $B$ already marked, ending the phase.
At this point the adversary lets $x_{2k-1} = u^\ast$ and $x_{2k} = v^\ast$, 
and (if $k < (n-1)/2$) starts phase $k+1$.

\smallskip

For any permutation $\pi$ of $V$ starting with $s$, through the rest of the proof we
denote by $D_\pi$ the variable D-values produced by $\algA$ when processing weight assignment $\weight_\pi$.
For each $k = 0,1,...,(n-1)/2$, the $(2k+1)$-permutation $x_0,x_1,...,x_{2k}$ chosen by
the adversary following the strategy above will be called the \emph{$k$th cruel prefix}.


\begin{figure}[t]
\begin{center}
\includegraphics[width = 4.5in]{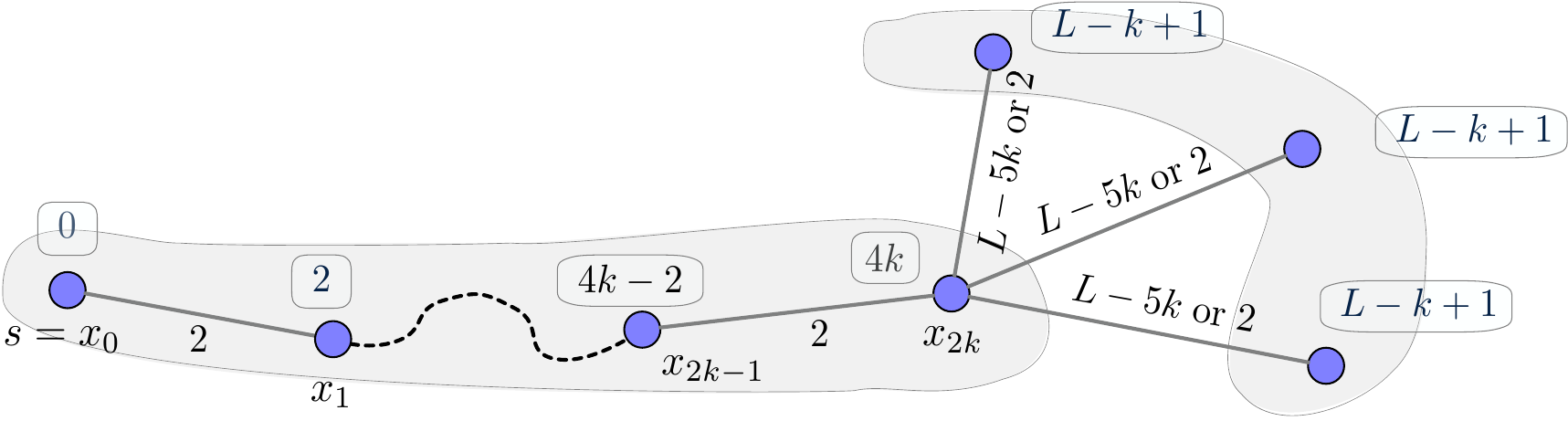}
\caption{The state of the game right after phase $k$ ends. Framed numbers next to vertices represent their D-values.
}\label{fig: relaxation proof invariant}
\end{center}
\end{figure}



\myparagraph{Invariant~(I)}.
We claim that the following invariant holds for each $k = 0,1,...,(n-1)/2$.
Let $x_0,x_1,...,x_{2k}$ be the adversary's $k$th cruel prefix.  Then for each permutation
$\pi$ starting with $x_0,x_1,...,x_{2k}$, the following properties hold when phase $k$ of the adversary strategy ends
(see Figure~\ref{fig: relaxation proof invariant} for illustration):
\begin{description}\setlength{\itemsep}{-0.03in}
	\item{(I0)} All adversary's answers to $\algA$'s queries in phases $1,2,...,k$ are correct for 
			weight assignment $\weight_\pi$.
	\item{(I1)} $D_\pi[x_j] = 2j$ for $j = 0,1,....,2k$.
	\item{(I2)} If $w\in V\setminus \braced{x_0,x_1,...,x_{2k}}$ then $D_\pi[w] =L-k+1$.
			
\end{description}

We postpone the proof of this invariant, and show first that it implies the $\Omega(n^3)$ lower bound.
Indeed, from Invariant~(I2) we conclude that
the D-values will not represent the correct distances until after the last step of phase $(n-1)/2$.
Since in each phase $k$ all edges in the set $A\cup B$ for phase $k$ will end up marked, 
the number of edge accesses in this phase is at
least $|A\cup B| = |A| + |A|(|A|-1) = |A|^2 = (n-2k+1)^2$. 
Thus, adding up the numbers of edge accesses in all phases $k = 1,2,...,(n-1)/2$, we obtain that
the total number of steps in algorithm~$\algA$ is at least 
$(n-1)^2 + (n-3)^2 + ... + 2^2 = \frac{1}{6}n(n^2-1) = \Omega(n^3)$, giving us the desired lower bound.

It remains to prove Invariant~(I). The invariant is true for $k=0$, by the way the D-values are initialized.
To argue that the invariant is preserved after each phase $k\ge 1$, we show that within this phase 
a more general invariant~(J) holds, below.


\myparagraph{Invariant~(J)}.
Let $\pi$ be a permutation with prefix $x_0,x_1,...,x_{2k}$, let
$X_{k-1} = x_0,x_1,...,x_{2k-2}$, and $Y_{k-1} = V\setminus X_{k-1}$.
We claim that the following properties are satisfied
during the phase, including right before and right after the phase.
\begin{description}\setlength{\itemsep}{-0.03in}
	\item{(J0)} All adversary's answers to $\algA$'s queries up to the current step are correct for $\weight_\pi$.
	\item{(J1)} $D_\pi[x_j] = 2j$ for $j = 0,1,...,2k-2$. 
	\item{(J2.1)} If $w\in Y_{k-1}\setminus\braced{u^\ast,v^\ast}$ then 
	\begin{minipage}[b]{3in}
	\begin{equation*}
			D_\pi[w] \;=\;
			\begin{cases} 
				L-k+2 &  \textrm{if} \; (x_{2k-2},w) \;\textrm{unmarked}
				\\
				L-k+1 & \textrm{if} \; (x_{2k-2},w) \;\textrm{marked}
			\end{cases}
	\end{equation*}
	\end{minipage}
	\item{(J2.2)} If $w = u^\ast$ then
	\begin{minipage}[b]{3in}
	\begin{equation*}
			D_\pi[u^\ast] \;=\;
			\begin{cases} 
				L-k+2 &  \textrm{if} \; (x_{2k-2},u^\ast) \;\textrm{unmarked}
				\\
				4k-2 & \textrm{if} \; (x_{2k-2},u^\ast) \;\textrm{marked}
			\end{cases}
	\end{equation*}
	\end{minipage}
	\item{(J2.3)} If $w = v^\ast$ then 
		\begin{minipage}[b]{3in}
	\begin{equation*}
			D_\pi[v^\ast] \; =\;
			\begin{cases} 
				L-k+2 &  \textrm{if} \; (x_{2k-2},v^\ast) \;\textrm{unmarked}
				\\
				L-k+1 & \textrm{if} \; (x_{2k-2},v^\ast) \;\textrm{marked and} \; (u^\ast,v^\ast) \; \textrm{unmarked}
				\\
				4k    & \textrm{if} \; (u^\ast,v^\ast) \;\textrm{marked}
			\end{cases}
	\end{equation*}
	\end{minipage}
\end{description}
When phase $k$ starts, these properties are identical to Invariant~(I) applied to the ending of phase $k-1$.
We now show that these invariants are preserved within a phase $k$. 
Assume that the invariants hold up to some step, and consider the next operation when $\algA$
accesses an edge $(u,v)$. If $w\neq v$ then $D_\pi[w]$ is not affected, so assume that $w = v$.
In the case analysis below, if the current step is a relaxation,
we will use notation $D_\pi[v]$ for the D-value at $v$
before this step and $D'_\pi[v]$ for the D-value at $v$ after the step.


\smallskip\noindent
\mycase{(s1)} $(u,v) = (x_{2k-2},v)\in A$. We consider separately the cases when this step is a relaxation or an edge query. 
\begin{description}\setlength{\itemsep}{-0.025in}
\item{\emph{Relaxation:}} 
Suppose first that $v\neq u^\ast$. Then we have $D_\pi[x_{2k-2}] + \weight(x_{2k-2},v) = (4k-4) + (L-5k+5) = L-k+1$.
Thus, using~(J2.1) and~(J2.3), if $(x_{2k-2},v)$ was already marked then nothing changes, and if $(x_{2k-2},v)$ wasn't marked
then $D'_\pi[v] = L-k+1$, preserving~(J2) because $(x_{2k-2},v)$ gets marked. (Note that in the special case
$v = v^\ast$, edge $(u^\ast,v^\ast)$ is not marked yet.)

For $v = u^\ast$ the argument is similar, except that now we
use~(J2.2): We have $D_\pi[x_{2k-2}] + \weight(x_{2k-2},u^\ast) = (4k-4) + 2 = 4k-2$,
so either $(x_{2k-2},u^\ast)$ is already marked and nothing changes, or $D'_\pi[u^\ast]=4k-2$ and $(x_{2k-2},u^\ast)$ gets marked.
\item{\emph{Edge query:}} 
The reasoning here is analogous to the case of relaxation above. If $v\neq u^\ast$, then
$D_\pi[x_{2k-2}] + \weight(x_{2k-2},v) = L-k+1$ and the correctness of the adversary's answers follows
from~(J2.1) and~(J2.3),
If $v = u^\ast$, then $D_\pi[x_{2k-2}] + \weight(x_{2k-2},u^\ast) = (4k-4) + 2 = 4k-2$,
and the correctness of the adversary's answers follows from~(J2.2).
\end{description}


\smallskip\noindent
\mycase{(s2)} $(u,v) \in B$. We consider separately the cases when this step is a relaxation or an edge query. 
\begin{description}\setlength{\itemsep}{-0.025in}
\item{\emph{Relaxation:}} Suppose first that $u\neq u^\ast$. Then $D_\pi[u]\ge L-k+1$, by~(J2.1) and~(J2.3).
(This is true  for the special case $u = v^\ast$, because $(u^\ast,v^\ast)$ is not yet marked.)
Since also $\weight(u,v)\ge 2$, this relaxation will not change the value of $D_\pi[v]$.

Next, consider the case $u = u^\ast$. If $(x_{2k-2},u^\ast)$ is unmarked then~(J2.2) also implies (as in the previous sub-case) that
the value of $D_\pi[v]$ will not change. 
So assume now that $(x_{2k-2},u^\ast)$ is marked, in which case $D_\pi[u^\ast] = 4k-2$.
If $v\neq v^\ast$ then the relaxation will not change the value of $D_\pi[v]$ because $\weight(u^\ast,v) = L$.
For $v = v^\ast$, we have $D'_\pi[v^\ast] = D_\pi[u^\ast] + \weight(u^\ast,v^\ast) = (4k-2) + 2 = 4k$,
preserving~(J2.3), because this relaxation will mark $(u^\ast,v^\ast)$.
\item{\emph{Edge query:}} If $(x_{2k-2},u)$ is not marked then, by (J2.1)-(J2.3) we have $D_\pi[u] = L-k+2$, 
and $\weight(u,v)\ge 2$, so the ``no'' answer by the adversary is correct.

Next, assume that $(x_{2k-2},u)$ is marked and $u\neq u^\ast$. Then $D_\pi[u] = L-k+1$, by
conditions~(J2.1) and~(J2.3) (since $(u^\ast,v^\ast)$ is still not marked). So in this case the
answer ``no'' is also correct. 

The final case is when $u = u^\ast$ and $(x_{2k-2},u^\ast)$ is marked, so $D_\pi[u^\ast] = 4k-2$.
Now, if $v\neq v^\ast$ then the adversary responds ``no'', and since $\weight(u^\ast,v)\ = L$, this is correct.
For $v = v^\ast$ we have $D_\pi[u^\ast] + \weight(u^\ast,v^\ast) = (4k-2) + 2 = 4k < D_\pi[v^\ast]$,
where the last inequality is true because $(u^\ast,v^\ast)$ is not marked.
So the ``yes'' answer is also correct. 
\end{description}


\smallskip\noindent
\mycase{(s3)} $(u,v)\notin A\cup B$. In this case we claim that $D_\pi[u] + \weight(u,v) \ge D_\pi[v]$,
which implies the correctness for both cases, when this operation is a relaxation and edge update.
The argument involves a few cases.

The first case is when $u\in Y_{k-1}$ and $v\in X_{k-1}$. Then we have $D_\pi[u] + \weight_{uv} \ge (4k-2) +2 > D_\pi[v]$, applying~(J1)-(J2.3). 

If $u\in X_{k-1}\setminus \braced{x_{2k-2}}$ and $v\in Y_{k-1}$ then there are two sub-cases,
and in both we apply~(J1) and~(J2.1)-(J2.3).
If $u = x_{2k-3}$ in which case $\weight(x_{2k-3},v) = L$,  the claim is trivial.
If $u = x_j$ for $j\le 2k-4$, then $D_\pi[x_j] = 2j$ and $\weight(x_j,v)\ge L-5j/2$,
so $D_\pi[u] + \weight(u,v) \ge (2j) + (L-5j/2) = L - j/2 \ge  L-k+2 \ge D_\pi[v]$.

The final case is when $u,v\in X_{k-1}$, say $u = x_i$ and $v = x_j$. Here we use condition~(J1).
If $i>j$ then $D_\pi[x_i] > D_\pi[x_j]$.
If $i< j-1$ then $\weight(x_i,x_j) \ge 2n$.
If $i = j-1$ then $D_\pi[x_{j-1}] = 2j-2$, $D_\pi[x_j] = 2j$ and $\weight(x_{j-1},x_j) = 2$.
In each of these sub-cases, the claim holds.

\smallskip\noindent

The case analysis above completes the proof of invariants~(J0)-(J2.3). By applying these invariants
to the end of the phase, when all edges in $A\cup B$ are marked, gives us that invariant~(I) holds after
the phase, as needed --- providing that the phase ends at all.

To complete the analysis of the adversary strategy we need to argue that phase $k$ must actually end,
in order for the D-values to represent correct distances from $s$. 
This follows directly from invariants~(J1)-(J2.3), because they imply that before the very last
step of the phase there is at least one vertex in $Y_{n-1}$ with D-value at least $L-k+1$,
which is larger than its distance from $s$.

It now only remains to remove the assumption that the D-values are initialized to $L+1$.
According to our model, they need to be initialized to edge lengths from $s$, which are:
$\ell(s,x_1) = 2$ and $\ell(s, v) = L$ for $v \ne x_1$ (since $s=x_0$). 
With this initialization, we only need to modify the first phase by marking all edges of the form $(x_0, v)$ immediately. 
Invariant~(J) then applies without further modification.


\section{Lower Bound for Deterministic Algorithms with Three Types of Queries}
\label{sec: deterministic algorithms three types of queries}


In this section, we prove Theorem~\ref{thm: all queries lower bound}(a), an $\Omega(n^3)$ lower
bound for deterministic algorithms using all three types of queries.
Our argument is essentially a reduction --- we show that any algorithm $\algA$ that uses 
D-queries, weight queries, edge queries and relaxation updates can be converted into an algorithm $\algB$
that has the same time complexity as $\algA$ and uses only edge queries and relaxations.
Our lower bound will then follow from Theorem~\ref{thm: relaxation-outcome lower bound}(a).

We start with some initial observations that, although not needed for the proof, contain some useful insights.
Since edge weights do not change, an algorithm can use weight queries to pre-sort all
edges in time $O(n^2 \log n)$, and then it doesn't need to make any more weight queries
during the computation.  This way, the algorithm's running time is not affected as long as it's at least $\Omega(n^2\log n)$.
Similarly, the algorithm can use D-queries to  maintain the total order of the D-values using, say, a binary search tree,
paying a small overhead of $O(\log n)$ for each update operation.
Then the algorithm's decisions at each step can as well depend on the total ordering of the 
vertices according to their current D-values. These changes will add at most an $O(\log n)$ factor to the running time.


\myparagraph{Potential-oblivious model.} 
Instead of working just with edge queries, we generalize our argument to \emph{potential-oblivious} query models.
We say that a query model $\modelQ$ is potential-oblivious if it satisfies the following
property for each weight assignment $\weight$ and potential $\phi$:
for any sequence of relaxations and queries from $\modelQ$
(with the D-values initialized as described in Section~\ref{sec: preliminaries}),
the outcomes of the queries for weight assignments $\weight$ and $\weight+\Delta\phi$ are the same.  
By routine induction, any algorithm using
a potential-oblivious model will perform the same sequence of queries and relaxations on assignments $\weight$ and $\weight+\Delta\phi$.
Also, it will compute the correct distances on $\weight$ if and only if it will compute them for $\weight+\Delta\phi$,
and in the same number of steps.
(To see this, note that for each vertex $v$ the invariant $D'[v] = D[v] + \phi(v)$ is preserved, where we use
notations $D$ and $D'$ to distinguish between the D-values in the computations for $\weight$ and $\weight'$.
The respective distances $\weight(s,v)$ and $\weight'(s,v)$ satisfy the same equation.)

For example, the edge query only model is potential-oblivious. Due to our initialization and
properties of relaxations, each value $D[v]$ always corresponds to the length of some path from $s$ to $v$.
Then the query is equivalent to comparing the length of two paths
with the same start and end points, and the query outcome is the same after adding $\Delta\phi$, by the distance preservation property.
Using these facts, potential-obliviousness follows from induction on the number of operations performed.


\myparagraph{Golomb-ruler potential.}
For our proof, we need a potential function $\phi$ for which in the induced weight assignment $\Delta\phi$
all edge weights are different. (This naturally implies that all values of $\phi$ are also different.)
Such an assignment is equivalent to a \emph{Golomb ruler} (also known as a Sidon set), 
which is a set of non-negative integers with unique pair-wise differences. 
A simple Golomb ruler can be constructed using fast growing sequences, such as $\braced{2^i-1}_{i=0}^{n-1}$,
but we are interested in sets contained in a small polynomial-in-$n$ range.
The asymptotic growth of Golomb rulers is well studied;
it is known that there are $n$-element Golomb rulers that are subsets of $\braced{1,2,...,N}$, for 
$N = n^2(1+o(1))$~\cite{1941_erdos_sidon,1938_singer_projective_geometry}, and that this bound on $N$ is essentially optimal.
Since the Golomb-ruler property is invariant under shifts, we can assume that a Golomb ruler contains number $0$.
For our purposes, this means that there exists a potential function $\phi$ that induces distinct edge weights with 
absolute maximum weight $O(n^2)$. (\cite{2005_bollobas_differences} shows that it is possible to obtain smaller maximum 
weights for certain classes of non-complete graphs, but this is not relevant to our constructions.)
We will call this function a \emph{Golomb-ruler potential}.


\begin{theorem}\label{thm: reduction theorem}
Let $\algA$ be a query/relaxation-based algorithm that uses relaxation updates, D-queries, weight queries and any queries from a 
potential-oblivious model $\modelQ$, and let $T(n)$ be the running time of $\algA$.
Then there is an algorithm $\algB$ with running time $O(T(n))$ that uses only relaxation updates and queries from $\modelQ$.
\end{theorem}

The idea of the proof is to convert a given weight assignment $\weight$ into another assignment $\weight'$
such that, if only queries from $\modelQ$ (and relaxations) are used, then 
(i) $\weight'$ is indistinguishable from $\weight$ using the queries from $\modelQ$,
and (ii) in $\weight'$ the ordering of weights and the ordering of all D-values are \emph{independent of $\weight$} and,
further, the ordering of the D-values is \emph{fixed throughout the computation}, even though the D-values themselves may vary.
$\algB$ can do this conversion ``internally'' and simulate $\algA$ on $\weight'$, and then it doesn't
need to make any D-queries and weight queries, because their outcomes are predetermined.

\begin{proof}
Let $\algA$ be a query/relaxation algorithm for that uses D-queries, weight queries, queries from $\modelQ$, and relaxation updates.
We construct $\algB$ that uses only queries from $\modelQ$ and relaxation updates. 
Let $\phi$ be the Golomb-ruler potential defined before the theorem. When run on a weight assignment $\weight$,
$\algB$ will internally simulate $\algA$ on weight assignment $\weight' = \weight + c\Delta\phi$, for $c = 2\maxedgeweight n + 1$.
We use notation $D'$ for the D-values computed by $\algA$.
The actions of $\algB$ depend on the execution of  $\algA$ on $\weight'$, as follows:
\begin{itemize}\setlength{\itemsep}{-0.02in}
\item When $\algA$ executes a weight query ``$\weight'_{uv} < \weight'_{xy} ? $'', $\algB$  directly executes the ``yes'' branch
from the query if $\Delta\phi(u,v) < \Delta\phi(x,y)$, or the ``no'' branch otherwise.
\item When $\algA$ executes a D-query ``$D'[u]<D'[v]?$'', then $\algB$ executes the ``yes'' branch if $\phi_u < \phi_v$, else it executes the ``no'' branch. 
\end{itemize}

This simulation can be more formally described as converting the decision tree of $\algA$ into the decision tree of $\algB$.
The tree of $\algB$ is obtained by splicing out each node $q$ representing a D-query or weight query.
This splicing consists of connecting the parent of $q$ to either the ``yes'' or ``no'' child of $q$,
determined by the appropriate inequality involving $\phi$, as explained above.

\smallskip

It remains to prove the correctness of $\algB$. We argue first
that $\algB$ will produce correct distances if run on $\weight'$ instead of $\weight$. For this,
we observe that $\weight'$ satisfies the following properties:
\begin{description}\setlength{\itemsep}{-0.02in}
\item{(p1)} For any two edges $e,f$, 
			we have $\weight'_e < \weight'_f$ if and only if $\Delta\phi(e) < \Delta\phi(f)$.

\item{(p2)} For any three vertices $u,x,y$, any $u$-to-$x$ path $P_x$ and any $u$-to-$y$ path $P_y$,
				we have $\weight'(P_x) < \weight'(P_y)$ if and only if $\phi_x < \phi_y$.

\end{description}
Indeed, both properties follow from the choice of $c$ and straighforward calculation.
For~(p1), $\weight'_e < \weight'_f$ if and only if	$\weight_e - \weight_f < c [\Delta\phi(f) - \Delta\phi(e)]$,
and because $|\weight_e - \weight_f| < c$ this inequality is determined by the sign of $\Delta\phi(f) - \Delta\phi(e)$, which is
always non-zero, by the Golomb-ruler property. (Note that here we only use that $c > 2\maxedgeweight$.)
The justification for~(p2) is similar: we have $\weight'(P_x) = \weight(P_x) + c(\phi_x - \phi_u)$ 
and  $\weight'(P_y) = \weight(P_y) + c(\phi_y - \phi_u)$, so
$\weight'(P_x) < \weight'(P_y)$ if and only if 
$\weight(P_x) - \weight(P_y) < c [\phi_y - \phi_x ]$, and since $|\weight(P_x) - \weight(P_y| < c$ this
inequality is determined by the sign of $\phi_y - \phi_x$.

Properties~(p1) and~(p2) imply that when we run $\algA$ on $\weight'$, in each weight query we
can equivalently use assignment $\Delta\phi$ instead of $\weight'$, and
instead of using a D-query we can compare the corresponding potential values. 
Therefore $\algB$ works correctly for $\weight'$.
But since now $\algB$ uses only relaxations and queries from $\modelQ$, that are
potential-oblivious, and $\weight'$ is obtained from $\weight$ by adding a weight
assignment induced by potential $c\phi$, $\algB$'s computation on $\weight$ will also be correct.
\end{proof}

Theorem~\ref{thm: reduction theorem}, together with Theorem~\ref{thm: relaxation-outcome lower bound}
implies the $\Omega(n^3)$ lower bound for query/update-based algorithms that use D-queries, weight queries,
any set of potential-oblivious queries, and relaxation updates.
Since the edge update is potential oblivious, Theorem~\ref{thm: all queries lower bound}(a) follows.

Further, using the construction from Theorem~\ref{thm: all queries lower bound}(a), where a weight assignment with
maximum weight $O(n)$ was used, the proof of Theorem~\ref{thm: reduction theorem} shows that
Theorem~\ref{thm: all queries lower bound}(a) holds even if all weights are bounded by $O(n^4)$. 


\myparagraph{A side result for general graphs.} 
The reduction in the proof of Theorem~\ref{thm: reduction theorem} extends naturally to arbitrary graphs. 
In particular, we can extend a result from Eppstein \cite{2023_Eppstein_non-adaptive_shortest_path}:

\begin{theorem}\label{thm: arbitrary graph}
	For any $n$ and $m$ where $n \le m \le n(n-1)$, 
	there exists a graph with $n$ nodes and $m$ edges where any deterministic algorithm $\algA$ using D-queries, weight queries, and relaxation updates 
	has worst-case running time $\Omega(nm/\log n)$. If $m = \Omega(n^{1+\varepsilon})$ for some $\varepsilon > 0$, the lower bound can be improved to $\Omega(nm)$.
\end{theorem}

\begin{proof}[Proof sketch]
	We focus on the case where $m$ is arbitrary. First note that the model using no queries and relaxation updates is equivalent to non-adaptive algorithms 
	(that is, universal relaxation sequences)
	described in \cite{2023_Eppstein_non-adaptive_shortest_path}. (It's also obvious that the query model using no queries is potential-oblivious.) 
	Let $G$ be the graph construction described in the proof of \cite[Theorem~3]{2023_Eppstein_non-adaptive_shortest_path}. 
	In particular, $G$ has $n$ nodes and $m$ edges, and for every non-adaptive algorithm on $G$, there is weight assignment that forces $\Omega(nm/\log n)$ relaxations. 
	If there exists an algorithm $\algA$ for $G$ using D-queries, weight queries, and relaxation updates that runs in $T(n)$ time, then by the same construction 
	as in Theorem~\ref{thm: reduction theorem}, there also exists an algorithm $\algB$ using only relaxation updates that runs in time $O(T(n))$. 
	Then an $o(nm/\log n)$ running time on $G$ would contradict the lower-bound on non-adaptive algorithms. The proof for the case $m = \Omega(n^{1+\varepsilon})$ is identical.
\end{proof}

As explained in Section~\ref{sec: lower bounds for randomized algorithms}, the reduction also applies to randomized algorithms, 
and because \cite{2023_Eppstein_non-adaptive_shortest_path} proves the same lower bounds for expected running time for randomized non-adaptive algorithms, 
the above bounds also apply to the randomized case.


\section{Lower Bounds for Randomized Algorithms}
\label{sec: lower bounds for randomized algorithms}

In this section, we extend the proofs in Sections~\ref{sec: deterministic algorithms relaxation queries}
and~\ref{sec: deterministic algorithms three types of queries} to obtain $\Omega(n^3)$ lower bounds
for randomized algorithms, proving Theorem~\ref{thm: relaxation-outcome lower bound}(b) and Theorem~\ref{thm: all queries lower bound}(b).
The proofs are based on Yao's principle~\cite{1977_yao_probabilistic_computations}: we give a
probability distribution on weight assignments for which the expectation of each deterministic algorithm's running time is $\Omega(n^3)$.

We fix the value of $n$. Let $\maxedgeweight$ be the maximum absolute value of weights used in the proof, whose value will
be specified later. Let $\weightset$ be the family of all weight assignments $\weight: E \to [-\maxedgeweight,\maxedgeweight]$.
Denote by $\algorithms$ the (finite) set of all deterministic query/relaxation-based algorithms with running time at most $2n^3$. 
We only need to consider algorithms in $\algorithms$, because any other algorithm in our model can be modified
to run in time at most  $2n^3$. To see why,
consider this algorithm's decision tree. For any node at depth $n^3$, replace its subtree by the Bellman-Ford
relaxation sequence. 
The resulting tree remains correct, and its depth is at most $2n^3$.

For a deterministic algorithm $\algA\in\algorithms$ and weight assignment $\weight \in \weightset$, denote by 
$\runtime(\calA,\weight)$ the running time of $\algA$ on assignment $\weight$. 
Let $\Pi(\algorithms)$ be the set of all probability distributions on $\algorithms$
and $\Pi(\weightset)$ be the set of all probability distributions on $\weightset$.
Any randomized algorithm $\randalgR$ is simply a probability distribution on $\algorithms$, so $\randalgR\in \Pi(\algorithms)$.
Denote by $\Exp_{x\sim\theta} f(x)$ the expected value of $f(x)$, for a random variable $x$ from distribution $\theta$.
The lemma below is a restatement of Yao's principle~\cite{1977_yao_probabilistic_computations} in our context:


\begin{lemma}\label{lem: yaos principle}
The following equality holds:
\begin{equation*}
	\min_{\randalgR\in \Pi(\algorithms)} \max_{\weight\in\weightset} \Exp_{\algA\sim\randalgR} \runtime(\algA,\weight)
				\;=\; \max_{\sigma\in\Pi(\weightset)} \min_{\algA\in\algorithms} \Exp_{\weight\sim\sigma} \runtime(\algA,\weight).
\end{equation*}
\end{lemma}

In this lemma, both sides involve the expected running time, with the difference being that 
on the left-hand side we consider randomized algorithms and their worst-case inputs,
while the right-hand side involves the probability distribution on \emph{input permutations} that is worst 
for \emph{deterministic algorithms}.


\emparagraph{Proof of Theorem~\ref{thm: relaxation-outcome lower bound}(b).}
(Sketch\footnote{A detailed proof will
appear in the full version of this paper.}.)
We give a probability distribution $\sigma$ on weight assignments $\weight$ for which every
deterministic algorithm needs $\Omega(n^3)$ steps in expectation to compute correct distances.
This is sufficient, as then Lemma~\ref{lem: yaos principle} implies that each
randomized algorithm $\randalgR$ makes  $\Omega(n^3)$ steps in expectation on some weight assignment.

Recall that in the proof of Theorem~\ref{thm: relaxation-outcome lower bound}(a) in Section~\ref{sec: deterministic algorithms relaxation queries}
we used weight assignments associated with permutations of vertices. This is also the case here, although this assignment needs to be modified.
For any permutation $\pi = x_0,x_1,...,x_{n-1}$ of the vertices, the corresponding weight assignment is
\begin{equation*}
\weight_\pi(x_i,x_j) \;=\; 
			\begin{cases}
						n	 	& \;\textrm{if}\; j = i+1
						\\
						L-(n+\tfrac{1}{2})i &\;\textrm{if}\; j \ge i+2 \;\textrm{and}\; i \;\textrm{is even}
						\\
						L	& \;\textrm{if}\; j \ge i+2 \;\textrm{and}\; i \;\textrm{is odd}
			\end{cases}
\end{equation*}
where $L$ is sufficiently large, say $5n^2$. (We explain later why larger weights are necessary.)
In our argument here, the adversary chooses the uniform distribution $\sigma$ on all $(n-1)!$ permutations $\pi$ starting with $s$.

In a certain sense, our goal now is simpler than in Section~\ref{sec: deterministic algorithms relaxation queries}, as the
adversary's job, which is to choose $\sigma$, is already done. We ``only'' need to lower bound the expected running time of
algorithms from  $\algorithms$ if the weights are distributed according to $\sigma$.
The challenge is that this argument needs to work for \emph{an arbitrary algorithm} from $\algorithms$.

So fix any deterministic algorithm $\algA \in \algorithms$.
We need to prove that $\Exp_{\weight\sim\sigma} \runtime(\algA,\weight) = \Omega(n^3)$.
A high-level approach in our proof is similar to the proof in Section~\ref{sec: deterministic algorithms relaxation queries}: we
partition the computation of $\algA$ into $(n-1)/2$ phases, and show that the expected length of each phase $k = 1,2,...,(n-1)/2$
is $\Omega((n-2k)^2)$. 

For a specific permutation $\pi = x_0,x_1,...,x_{n-1}$ with $x_0 = s$ and $k = 1,2,...,(n-1)/2$, 
let $t_k(\pi)$ be the first time step such that in steps $1,2,...,t_k(\pi)$ the edges
$(x_0,x_1), ..., (x_{2k-1},x_{2k})$ have been accessed by $\algA$ (that is, relaxed or queried) in this particular order.
We refer to the time interval $(t_{k-1}(\pi),t_k(\pi)]$ as \emph{phase $k$ for permutation $\pi$}.


The proof idea is this: In each phase $k$ of the strategy in Section~\ref{sec: deterministic algorithms relaxation queries} the adversary
was able to force the algorithm to relax all edges from $x_{2k-2}$ to $Y$ before revealing edge $(x_{2k-1},x_{2k})$, thus
ensuring that at all times all D-values differ at most by $1$.  This is not possible anymore,
because now the algorithm can get ``lucky'' and relax edges $(x_{2k-2},x_{2k-1})$ and $(x_{2k-1},x_{2k})$ before all edges 
$(x_{2k-2},u)$ for $u\in Y$ are
relaxed, and then the D-values for such vertices $u$ will reflect the relaxations that occurred in some earlier phases.
But we can still bound the differences between D-values. Namely,
by our choice of the length function above, any two D-values will differ by at most $n/2$. 
Thus, since all edge lengths are at least $n$, the negative answers to all edge queries inside $Y$ are still correct,
independently of the suffix $x_{2k-1},...,x_{n-1}$ of $\pi$.

More specifically, the analysis is based on establishing two invariants, captured by the claim below (formal proof omitted here).

\begin{claim}\label{cla: randomized invariant}
The following invariants are satisfied when each phase $k$ starts:
\begin{description}\setlength{\itemsep}{-0.025in}
\item{(R1)} The computation of $\algA$ up until phase $k$ starts is
		independent of the suffix $x_{2k-1}, x_{2k},...,x_{n-1}$ of $\pi$.
\item{(R2)} The D-values have the following form:
		$D[x_j] = jn$ for $j\le 2k-2$,
		and $D[x_j] \in [L-k+1,L]$ for $j\ge 2k-1$.
\end{description}
\end{claim}

Next, define $\tildet_k$ to be a random variable whose values are $t_k(\pi)$ for permutations $\pi$ distributed
randomly according to $\sigma$.
We refer to the time interval $(\tildet_{k-1},\tildet_k]$ as \emph{phase $k$}, and
let $\phaselength_k = \tildet_k - {\tildet_{k-1}}$ be the random variable equal to the length of this phase.

We then prove the following claim:

\begin{claim}\label{cla: random phase length}
$\Exp_{\weight\sim\sigma}[\phaselength_k] \ge  \half(n-2k+1)(n-2k+2)$.
\end{claim}

Let $\barz = z_0,z_1,...z_{2k-2}$ be some fixed $(2k-1)$-permutation of $V$ with $z_0 = s$.
Let $H$ be the event that $\pi$ starts with $\barz$. It is sufficient to prove the inequality in Claim~\ref{cla: random phase length}
for the conditional expectation $\Exp_{\weight\sim\sigma}[\phaselength_k|H]$.

So assume that event $H$ is true. Let $Y = V \setminus \braced{z_0,...,z_{2k-2}}$.
Now the argument is this: The suffix $x_{2k-1},...,x_{n-1}$ of $\pi$ is a random permutation of $Y$
and the edge $(x_{2k-1},x_{2k})$ is uniformly distributed among the edges in $Y$. 
Algorithm~$\algA$ is deterministic and all edge queries for edges inside $Y$, except for edge
$(x_{2k-1},x_{2k})$ (and only if $(x_{2k-2},x_{2k-1})$ has already been relaxed), 
will have negative answers. Similarly, all queries to edges from $x_{2k-2}$ to $Y$ will have positive answers.
So $\algA$ will be accessing these edges in some order that is
uniquely determined by the state of $\algA$ when phase $k$ starts.
Since there are $(n-2k+1)(n-2k+2)$ edges in $Y$, this implies that on average it will take
$\half(n-2k+1)(n-2k+2)$ steps for $\algA$ to access $(x_{2k-1},x_{2k})$,
even if we don't take into account that $(x_{2k-2},x_{2k-1})$ needs to be accessed first.
This will imply Claim~\ref{cla: random phase length}.

\smallskip

We now continue the proof of Theorem~\ref{thm: relaxation-outcome lower bound}(b).
For the algorithm to be correct, if the chosen permutation $\pi$ is $x_0,x_1,...,x_{n-1}$,
then the algorithm needs to relax the edges on this path in order as they appear on
the path. So its running time is at least $t_{(n-1)/2}(\pi)$.
Since $\tildet_{(n-1)/2} = \sum_{k=1}^{(n-1)/2} \phaselength_k$, using Claim~\ref{cla: random phase length}
and applying the linearity of expectation we obtain that
$\Exp_{\weight\sim\sigma} \runtime(\algA,\weight) 
	\ge \Exp_{\weight\sim\sigma}[\tildet_{(n-1)/2}]
	= \sum_{k=1}^{(n-1)/2} \Exp_{\weight\sim\sigma}[\phaselength_k]
	= \Omega(n^3)$,
completing the proof.


\emparagraph{Proof of Theorem~\ref{thm: all queries lower bound}(b).}
There is not much to prove here, because the reduction described in
Section~\ref{sec: deterministic algorithms three types of queries} applies
with virtually no changes to randomized algorithms.
Indeed, just like in the proof of Theorem~\ref{thm: all queries lower bound}(a)
(or more specifically the proof of Theorem~\ref{thm: reduction theorem}),
suppose that $\algR$ is a randomized algorithm that uses all three types of queries:
$D$-queries, weight-queries, the queries from model $\modelQ$, as well as relaxation updates,
and let $T(n)$ be $\algR$'s expected running time. We can convert $\algR$
into a randomized algorithm $\algR'$ with running time $O(T(n))$
that uses only the queries from $\modelQ$ and relaxation updates.
With this, Theorem~\ref{thm: all queries lower bound}(b)
follows from Theorem~\ref{thm: relaxation-outcome lower bound}(b).


\emparagraph{High-probability bounds.}
Using standard reasoning (see~\cite{2023_Eppstein_non-adaptive_shortest_path}, for example),
our lower bound results for expectation imply respective high-probability bounds,
namely that there are no randomized algorithms in the models from
Theorems~\ref{thm: relaxation-outcome lower bound} and~\ref{thm: all queries lower bound}
that compute correct distance values in time $o(n^3)$ with probability at least $1-o(1)$.

To justify this, suppose that $\algR$ is a randomized algorithm that computes correct distance values
in time $T(n) = o(n^3)$ with probability $1-o(1)$.
Consider the algorithm $\algR'$ obtained from $\algR$ by
switching to the Bellman-Ford relaxation sequence right after step $T(n)$.
The expected running time of $\algR'$ is then at most
$T(n) + o(1) n^3 = o(n^3)$, but this would contradict our
lower bounds in Theorems~\ref{thm: relaxation-outcome lower bound}(b) and~\ref{thm: all queries lower bound}(b).


\section{Final Comments and Open Problems}
\label{sec: final comments}

Our reduction in Section~\ref{sec: deterministic algorithms three types of queries}
introduces negative weights, raising a natural question: \emph{Is it possible to use $o(n^3)$ relaxations
with only weight-queries for instances with non-negative weights?}
(This question is of purely theoretical interest, because $O(n^2)$ relaxations can be achieved, using Dijkstra's algorithm, if
D-queries are used instead.)
Our proof techniques do not work for this variant. 
The reason is, in the instances we construct the shortest-path tree is a Hamiltonian path, and for such instances
this path can be uniquely determined by the weight ordering: start from $s$, and at each step follow 
the shortest outgoing edge from the current vertex to a yet non-visited vertex. 
So only $n-1$ relaxations are needed. It is unclear what is the ``hard'' weight ordering in this case. It can be shown
that in the two extreme cases: 
(i) if the weight orderings of outgoing edges from each vertex are agreeable (that is, they are determined
by a permutation of the vertices),  or 
(ii) if they are random, then there is a relaxation sequence of length only $O(n^{2.5})$ (and this likely can be improved further).

A natural extension of our query/relaxation model would be to allow \emph{unconditional edge updates}
of the form $D[v]\assign D[u] + \weight_{uv}$.
A combination of such edge updates and D-queries allows an algorithm to check for properties
that are impossible to test if only relaxation updates are used. For example,
by applying edge updates repeatedly around cycles, such an algorithm would be able to
determine, for any given rational number $c$, whether one cycle is at least $c$ times longer than some other cycle. 

A more open-ended question  is to determine if there are simple types of queries, say some
linear inequalities involving weights and the D-values (with a constant number of variables),
that would be sufficient to yield an adaptive algorithm (possibly randomized) that makes $o(n^3)$ relaxations.

The case of random universal sequences is also not fully resolved. While it is known that the asymptotic
bound is $\Theta(n^3)$, there is a factor-of-$4$ gap for the leading constant, between $\onetwelth$ and $\onethird$
\cite{2023_Eppstein_non-adaptive_shortest_path,2012_bannister_eppstein_randomized_speedup}.

We remark that our proofs are somewhat sensitive to the initialization of the D-values. Recall that
in our model we assume that initially $D[v] = \weight_{sv}$ for $v\neq s$. This is natural,
and it guarantees that at all times the D-values represent lengths of paths from $s$. 
It also has the property of being language- and platform-independent. However, some descriptions of
shortest-path algorithms initialize the D-values to infinity, or some very large number.
The proof of Theorem~\ref{thm: relaxation-outcome lower bound} in Section~\ref{sec: deterministic algorithms relaxation queries}
can be modified to work if
the D-values were initialized to some sufficiently large value $M$ (the adversary can then use
$L = M-1$ in her strategy). However, then the edge lengths are no longer polynomial, and 
the proof of Theorem~\ref{thm: all queries lower bound} in Section~\ref{sec: deterministic algorithms three types of queries}
does not apply in its current form.
Initializing to infinity would also affect the proofs.
The reduction in Section~\ref{sec: deterministic algorithms three types of queries} can be modified to account for
infinite D-values,  but we don't know how to adapt the proof in Section~\ref{sec: deterministic algorithms relaxation queries}
to this model. We leave open the problem of finding a more ``robust'' lower bound proof, that works for an arbitrary
valid initialization and uses only polynomial weights.


\bibliographystyle{plain}
\bibliography{shortest_paths.bib}

\end{document}